\documentclass[preprint]{sigplanconf}

\usepackage{amsmath,amssymb,enumerate}
\usepackage{multicol}
\usepackage{verbatim}

\usepackage{graphics}
\usepackage{graphicx}
\usepackage{subfigure}

\newtheorem{theorem}{Theorem}%[section]
\newtheorem{lemma}[theorem]{Lemma}

\newtheorem{definition}[theorem]{Definition}
\newtheorem{property}[theorem]{Property}

\newcommand{\sq}{\hbox{\rlap{$\sqcap$}$\sqcup$}}
\newcommand{\qed}{\hspace*{\fill}\sq}

\newenvironment{proof}{\noindent {\bf Proof.}\ }{\qed\par\vskip 4mm\par}

\begin{document}

\title{Well-Structured Futures and Cache Locality}
%\subtitle{Subtitle Text, if any}

\authorinfo{Maurice Herlihy}
           {Computer Science Department \\
           Brown University}
           {mph@cs.brown.edu}
\authorinfo{Zhiyu Liu}
           {Computer Science Department \\
           Brown University}
           {zhiyu\_liu@brown.edu}

\maketitle

\begin{abstract}
In \emph{fork-join parallelism},
a sequential program is split into a directed acyclic graph of tasks
linked by directed dependency edges,
and the tasks are executed,
possibly in parallel,
in an order consistent with their dependencies.
A popular and effective way to extend fork-join parallelism
is to allow threads to create \emph{futures}.
A thread creates a future to hold the results of a computation,
which may or may not be executed in parallel.
That result is returned when some thread \emph{touches} that future,
blocking if necessary until the result is ready.

Recent research has shown that while futures can, of course,
enhance parallelism in a structured way,
they can have a deleterious effect on cache locality.
In the worst case,
futures can incur $\Omega(P T_\infty + t T_\infty)$ deviations, which implies $\Omega(C P T_\infty + C t T_\infty)$ additional cache misses,
where $C$ is the number of cache lines,
$P$ is the number of processors,
$t$ is the number of touches,
and $T_\infty$ is the \emph{computation span}.
Since cache locality has a large impact on software performance on modern
multicores, this result is troubling.

In this paper, however,
we show that if futures are used in a simple, disciplined way,
then the situation is much better:
if each future is touched only once,
either by the thread that created it,
or by a thread to which the future has been passed from the thread that created it,
then parallel executions with work stealing can incur at most
$O(C P T^2_\infty)$ additional cache misses,
a substantial improvement.
This structured use of futures is characteristic of many
(but not all) parallel applications.
\end{abstract}

\section{Introduction}
Futures \cite{Halstead84,Halstead85} are an attractive way to structure many
parallel programs because they are easy to reason about
(especially if the futures have no side-effects) and they lend themselves well
to sophisticated dynamic scheduling algorithms,
such as work-stealing \cite{Blumofe99} and its variations,
that ensure high processor utilization.
At the same time, however,
modern multicore architectures employ complex multi-level memory hierarchies,
and technology trends are increasing the relative performance differences among
the various levels of memory.
As a result,
processor utilization can no longer be the sole figure of merit for schedulers.
Instead, the \emph{cache locality} of the parallel execution will become
increasingly critical to overall performance.
As a result, cache locality will increasingly join processor utilization
as a criterion for evaluating dynamic scheduling algorithms.

Several researchers \cite{Acar00,Spoonhower09} have shown, however,
that introducing parallelism through the use of futures can
sometimes substantially reduce cache locality.
In the worst case,
if we add futures to a sequential program,
a parallel execution managed by a work-stealing scheduler can incur $\Omega(P T_\infty + t T_\infty)$ deviations,
which, as we show, implies
$\Omega(C P T_\infty + C t T_\infty)$ more cache misses than the sequential
execution.
Here, $C$ is the number of cache lines,
$P$ is the number of processors,
$t$ is the number of touches,
and $T_\infty$ is the computation's \emph{span} (or \emph{critical path}).
As technology trends cause the cost of cache misses to increase,
this additional cost is troubling.

This paper makes the following three contributions.
First, we show that if futures are used in a simple, disciplined way,
then the situation with respect to cache locality is much better:
if each future is touched only once,
either by the thread that created it,
or by a thread to which the future has been passed directly or indirectly
from the thread that created it,
then parallel executions with work stealing can incur at most
$O(C P T^2_\infty)$ additional cache misses,
a substantial improvement over the unstructured case.
This result provides a simple way to identify computations for which
introducing futures will not incur a high cost in cache locality,
as well as providing guidelines for the design of future parallel computations.
(Informally, we think these guidelines are natural,
and correspond to structures programmers are likely to use anyway.)
We also prove that this upper bound is tight within a factor of $C$.
Our second contribution is to observe that when the scheduler has a choice
between running the thread that created a future,
and the thread that implements the future,
running the future thread first provides better cache locality.
Finally, we show that certain variations of structured computation also
have good cache locality.

The paper is organized as follows.
Section~\ref{section:model} describes the model for
future-parallel computations.
In Section~\ref{section:Work-Stealing and Cache Locality},
we describe parsimonious work-stealing schedulers,
and briefly discuss their cache performance measures.
In Section~\ref{section_structured},
we define some restricted forms of structured future-parallel computations.
Among them, we highlight structured single-touch computations,
which, we believe, are likely to arise naturally in many programs.
In Section~\ref{section_futurethreadfirst},
we prove that work-stealing schedulers on structured single-touch computations incur only
$O(C P T^2_\infty)$ additional cache misses, if a processor always chooses the
future to execute first when it creates that future.
We also prove this bound is tight within a factor of $C$.
In section~\ref{section:Parent Thread First at Each Fork},
we show that if a processor chooses the current thread over the
future thread when it creates that future,
then the cache locality of a structured single-touch computation can be much worse.
In Section~\ref{section:Other Kinds of Structured Computations},
we show that some other kinds of structured future-parallel computations also
achieve relatively good cache locality.
Finally, we present conclusions in Section~\ref{section:Conclusions}.

\section{Model}
\label{section:model}
In \emph{fork-join parallelism}~\cite{Blelloch95,Blelloch96,R.Blumofe96},
a sequential program is split into a directed acyclic graph of \emph{tasks}
linked by directed dependency edges.
These tasks are executed in an order consistent with their dependencies,
and tasks unrelated by dependencies can be executed in parallel.
Fork-join parallelism is well-suited to dynamic load-balancing techniques such
as \emph{work stealing}~\cite{Burton81,Halstead84,Arora98,Blumofe99,Acar00,Halstead85,Blumofe95,Frigo98,Kranz89,Agrawal07,Chase05}.

A popular and effective way to extend fork-join parallelism
is to allow threads to create \emph{futures}~\cite{Halstead84,Halstead85,Arvind89,Blelloch97,Fluet08}.
A future is a data object that represents a \emph{promise} to deliver the
result of an asynchronous computation when it is ready.
That result becomes available to a thread when the thread \emph{touches} that
future, blocking if necessary until the result is ready.
Futures are attractive because they provide greater flexibility than fork-join
programs,
and they can also be implemented effectively using dynamic load-balancing
techniques such as work stealing.
Fork-join parallelism can be viewed as a special case of
future-parallelism,
where the \texttt{spawn} operation is an implicit future creation,
and the \texttt{sync} operation is an implicit touch of the untouched futures
created by a thread.
Future-parallelism is more flexible than fork-join parallelism,
because the programmer has finer-grained control over touches (joins).

\subsection{Computation DAG}
A thread creates a future by marking an expression (usually a method call) as a
\emph{future}.
This statement spawns a new thread to evaluate that expression in parallel with
the thread that created the future.
When a thread needs access to the results of the computation,
it applies a \emph{touch} operation to the future.
If the result is ready, it is returned by the touch,
and otherwise the touching thread blocks until the result becomes ready.
Without loss of generality,
we will consider fork-join parallelism to be a special case of
future-parallelism, where forking a thread creates a future,
and joining one thread to another is a touch operation.

Our notation and terminology follow earlier
work~\cite{Arora98,Blumofe99,Acar00,Spoonhower09}.
A future-parallel computation is modeled as a \emph{directed acyclic graph} (DAG).
Each node in the DAG represents a task (one or more instructions),
and an edge from node $u$ to node $v$ represents the dependency constraint
that $u$ must be executed before $v$.
We follow the convention that each node in the DAG has in-degree and out-degree
either 1 or 2,
except for a distinguished \emph{root node} with in-degree 0,
where the computation starts,
and a distinguished \emph{final node} with out-degree 0,
where the computation ends.

There are three types of edges:
\begin{itemize}
\item
\emph{continuation edges},
which point from one node to the next in the same thread,

\item
\emph{future edges} (sometimes called \emph{spawn} edges),
which point from node $u$ to the first node of another thread spawned at
$u$ by a future creation,

\item
\emph{touch edges} (sometimes called \emph{join} edges),
directed from a node $u$ in one thread $t$ to a node $v$ in another thread,
indicating that $v$ touches the future computed by $t$.
\end{itemize}
A \emph{thread} is a maximal chain of nodes connected by continuation edges.
There is a distinguished \emph{main thread} that begins at the root node and
ends at the final node,
and every other thread $t$ begins at a node with an
incoming future edge from a node of the thread that spawns $t$.
The last node of $t$ has only one outgoing edge which is a touch edge directed to another thread,
while other nodes of $t$ may or may not have incoming and outgoing touch edges.
A \emph{critical path} of a DAG is a longest directed path in the DAG,
and the DAG's \emph{computation span} is the length of a critical path.

\begin{figure}
\centering
\subfigure[]{\includegraphics[width=5.6cm]{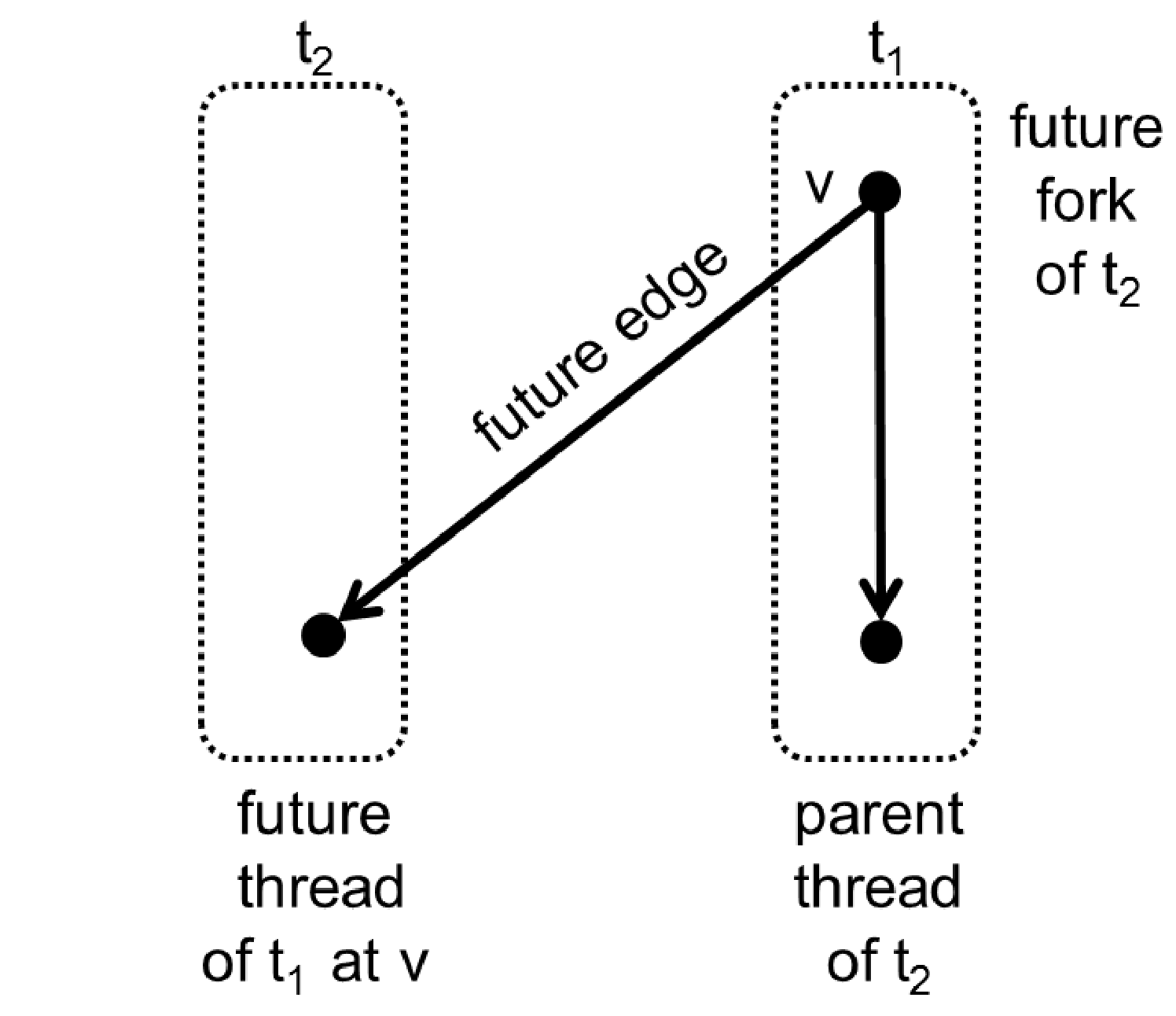}}
\subfigure[]{\includegraphics[width=5.6cm]{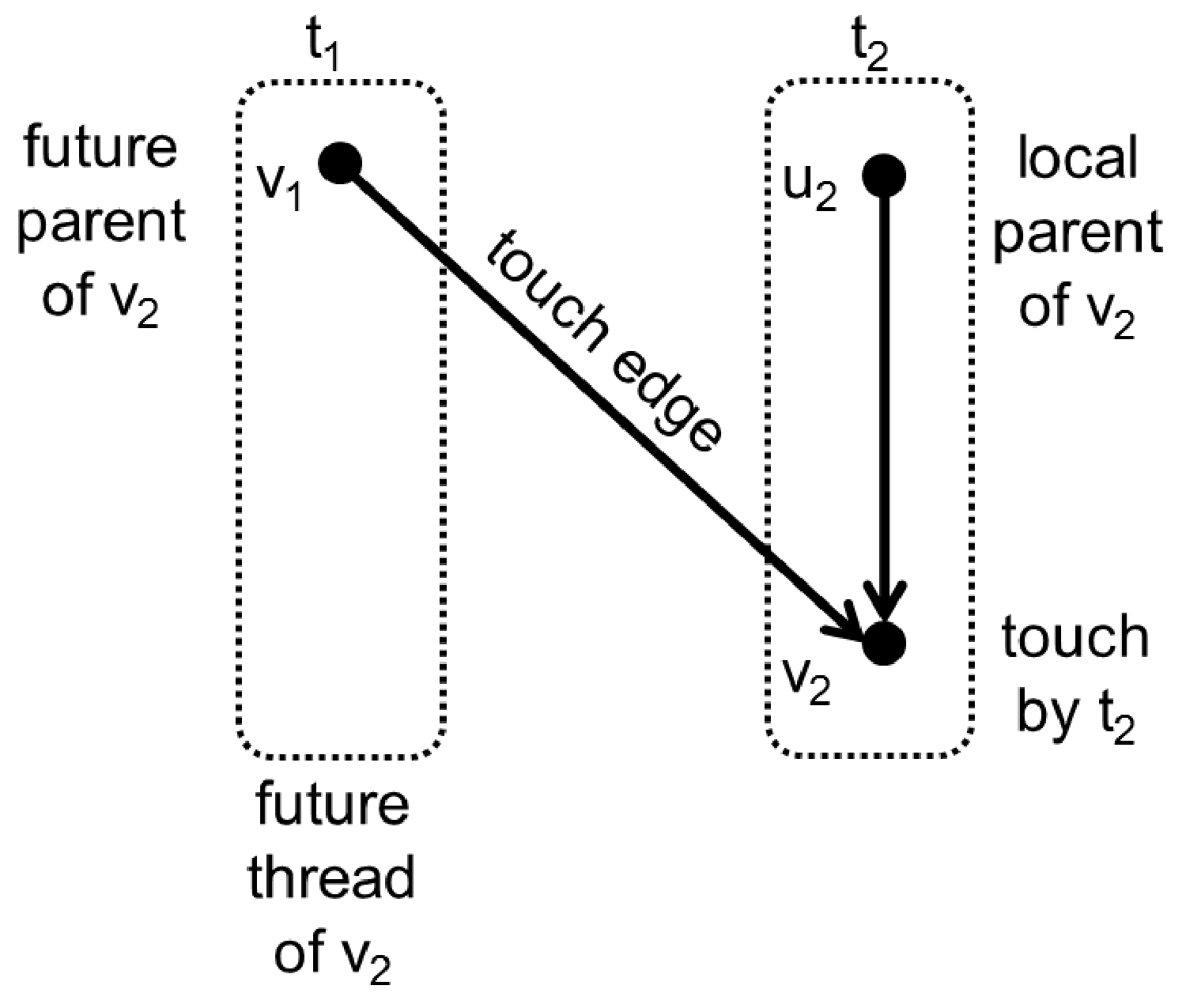}}
\caption{Node and thread terminology}
\label{figure:terminology}
\end{figure}
As illustrated in Figure~\ref{figure:terminology},
if a thread $t_1$ spawns a new thread $t_2$ at node $v$ in $t_1$ (i.e., $v$ has two out-going edges, a continuation edge and a future edge to the first node of $t_2$),
then we call $t_1$ the \emph{parent thread} of $t_2$,
$t_2$ the \emph{future thread} (of $t_1$) at $v$,
and $v$ the \emph{fork} of $t_2$.
A thread $t_3$ is a \emph{descendant thread} of $t_1$ if $t_3$ is a future
thread of $t_1$ or, by induction,
$t_3$'s parent thread is a descendant thread of $t_1$.

If there is a touch edge directed from node $v_1$ in thread $t_1$
to node $v_2$ in thread $t_2$ (i.e., $t_2$ touches a future computed by $t_1$),
and a continuation edge directed from node $u_2$ in $t_2$ to $v_2$,
then we call node $v_2$ a \emph{touch of} $t_1$ by $t_2$,
$v_1$ the \emph{future parent} of $v_2$,
$u_2$ the \emph{local parent} of $v_2$,
and $t_1$ the future thread of $v_2$. (Note that the touch $v_2$ is actually a node in thread $t_2$.)
We call the fork of $t_1$ the \emph{corresponding fork} of $v_2$.

Note that only touch nodes have in-degree 2.
To distinguish between the two types of nodes with out-degree 2,
forks and future parents of touches,
we follow the convention of previous work that the children of a fork
both have in-degree 1 and cannot be touches.
In this way,
a fork node has two children with in-degree 1,
while a touch's future parent has a (touch) child with in-degree 2.

We follow the convention that when a fork appears in a DAG,
the future thread is shown on the left, and the future parent on the right.
(Note that this does not mean the future thread is chosen to execute first at a fork.)
Similarly, the future parent of a touch is shown on the left,
and the local parent on the right.

We use the following (standard) notation.
Given a computation DAG,
$P$ is the number of processors executing the computation,
$t$ is the number of touches in the DAG,
$T_{\infty}$, the \emph{computation span} (or \emph{critical path}),
is the length of the longest directed path,
and $C$ is the number of cache lines in each processor.

\section{Work-Stealing and Cache Locality}
\label{section:Work-Stealing and Cache Locality}
In the paper, we focus on parsimonious work stealing algorithms \cite{Arora98},
which have been extensively studied~\cite{Arora98,Blumofe99,Acar00,Spoonhower09,Blumofe98}
and used in systems such as Cilk~\cite{Blumofe95}.
In a parsimonious work stealing algorithm,
each processor is assigned a double-ended queue (deque).
After a processor executes a node with out-degree 1,
it continues to execute the next node if the next node is ready to execute.
After the processor executes a fork,
it pushes one child of the fork onto the bottom of its deque and executes the other.
When the processor runs out of nodes to execute,
it pops the first node from the bottom of its deque if the deque is not empty.
If, however, its deque is empty,
it steals a node from the top of the deque of an arbitrary processor.

In our model, a cache is fully associative and consists of multiple \emph{cache lines},
each of which holds the data in a \emph{memory block}.
Each instruction can access only one memory block.
In our analysis we focus only on the widely-used \emph{least-recently used}
(LRU) cache replacement policy, but our results about the upper bounds on cache
overheads should apply to all \emph{simple} cache replacement policies~\cite{Acar00}.
\footnote{That is because the upper bounds in this paper are based on the results
of \cite{Acar00} that bound the number of drifted nodes (i.e., deviations), and those
results hold for all simple cache replacement policies,
even with set associative caches, as discussed in \cite{Acar00}.}

The \emph{cache locality} of an execution is measured by the number of cache
misses it incurs,
which depends on the structure of the computation.
To measure the effect on cache locality of parallelism,
it is common to compare cache misses encountered in a sequential
execution to the cache misses encountered in various parallel executions,
focusing on the number of \emph{additional} cache misses
introduced by parallelism.

Scheduling choices at forks affect the cache locality of executions with work stealing.
After executing a fork,
a processor picks one of the two child nodes to execute and pushes the other into its deque.
For a sequential execution,
whether a choice results in a better cache performance is a characteristic of
the computation itself.
For a parallel execution of a computation satisfying certain properties,
however,
we will show that choosing future threads (the left children) at forks
to execute first guarantees a relatively good upper bound on the number of
additional cache misses,
compared to a sequential execution that also chooses future threads first.
In contrast, choosing the parent threads (the right children) to execute first
can result in a large number of additional cache misses,
compared to a sequential execution that also chooses parent threads first.

\section{Structured Computations}
\label{section_structured}
Consider a sequential execution where node $v_1$ is executed immediately before
node $v_2$.
A \emph{deviation}~\cite{Spoonhower09}, also called a drifted node \cite{Acar00}, occurs in a parallel execution if
a processor $P$ executes $v_2$, but not immediately after $v_1$.
For example,
$p$ might execute $v_1$ after $v_2$,
it might execute other nodes between $v_1$ and $v_2$,
or $v_1$ and $v_2$ might be executed by distinct processors.

\cite{Spoonhower09} showed that a parallel execution of a
future-parallel computation with work stealing can incur $\Omega(PT_{\infty}+ tT_{\infty})$ deviations. This implies a parallel execution of a
future-parallel computation with work stealing can incur
$\Omega(PT_{\infty}+ tT_{\infty})$ additional cache misses.
With minor modifications in that computation (see Figure~\ref{figure:worstcaseDAGinSpoonhower09}),
a parallel execution can even incur $\Omega(CPT_{\infty} + CtT_{\infty})$
additional cache misses.

\begin{figure}
\centering
\includegraphics[width=5cm]{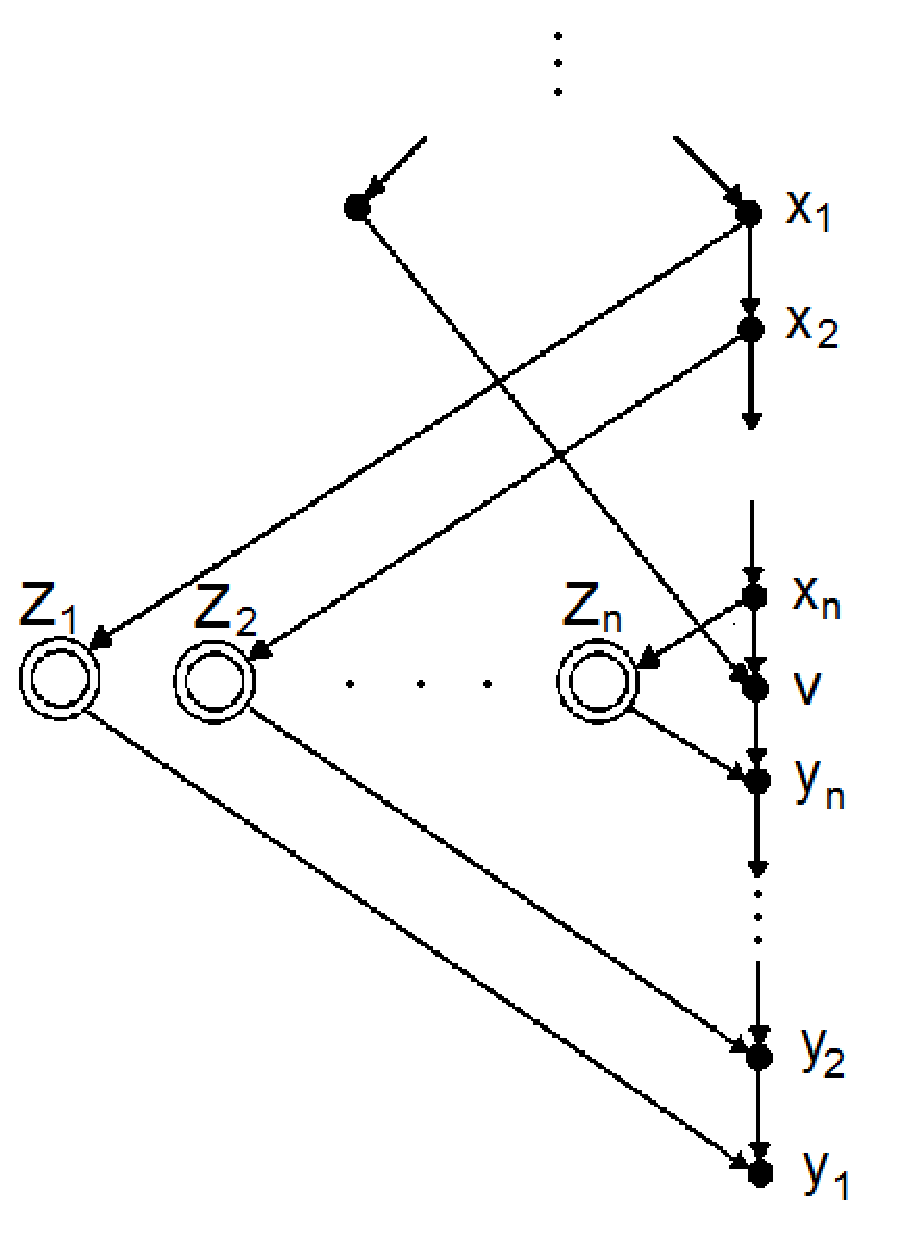}
\caption{The interesting part of the bound is $\Omega(CtT_{\infty})$. Figure 5 in \cite{Spoonhower09} shows a DAG, as a building block of a worst-case computation, that can incur $\Omega(T_{\infty})$ deviations because of one touch. We can replace it with the DAG in Figure~\ref{figure:worstcaseDAGinSpoonhower09}, which can incur $\Omega(CT_{\infty})$ additional cache misses due to one touch $v$ (if the processor at a fork always chooses the parent thread to execute first), so that the worst-case computation in \cite{Spoonhower09} can incur $\Omega(CtT_{\infty})$ additional cache misses because of $t$ such touches. This DAG is similar to the DAG in Figure~\ref{fig_parentfirst_lowerbound1}(a) in this paper. The proof of Theorem~\ref{parentthreadfirst} shows how a parallel execution of this DAG incurs $\Omega(CT_{\infty})$ additional cache misses.}
\label{figure:worstcaseDAGinSpoonhower09}
\end{figure}

Our contribution in this paper is based on the observation that such poor cache
locality occurs primarily when futures in the DAG can be touched by arbitrary
threads, resulting in unrealistic and complicated dependencies.
For example, in the worst-case DAGs in \cite{Spoonhower09} that can incur
significantly high cache overheads, futures are touched by threads that can be
created before the future threads computing these futures were created.
As illustrated in Figure~\ref{fig_unstructured},
a parallel execution of such a computation can arrive at a scenario where a
thread touches a future before the future thread computing that future has
been spawned.
(As a practical matter,
an implementation must ensure that such a touch does not return a reference to a
memory location that has not yet been allocated.)
Such scenarios are avoided by \emph{structured}
future-parallel computations (e.g. Figure~\ref{fig_structured})
that follow certain simple restrictions.

\begin{figure}
\centering
\includegraphics[width=7.5cm]{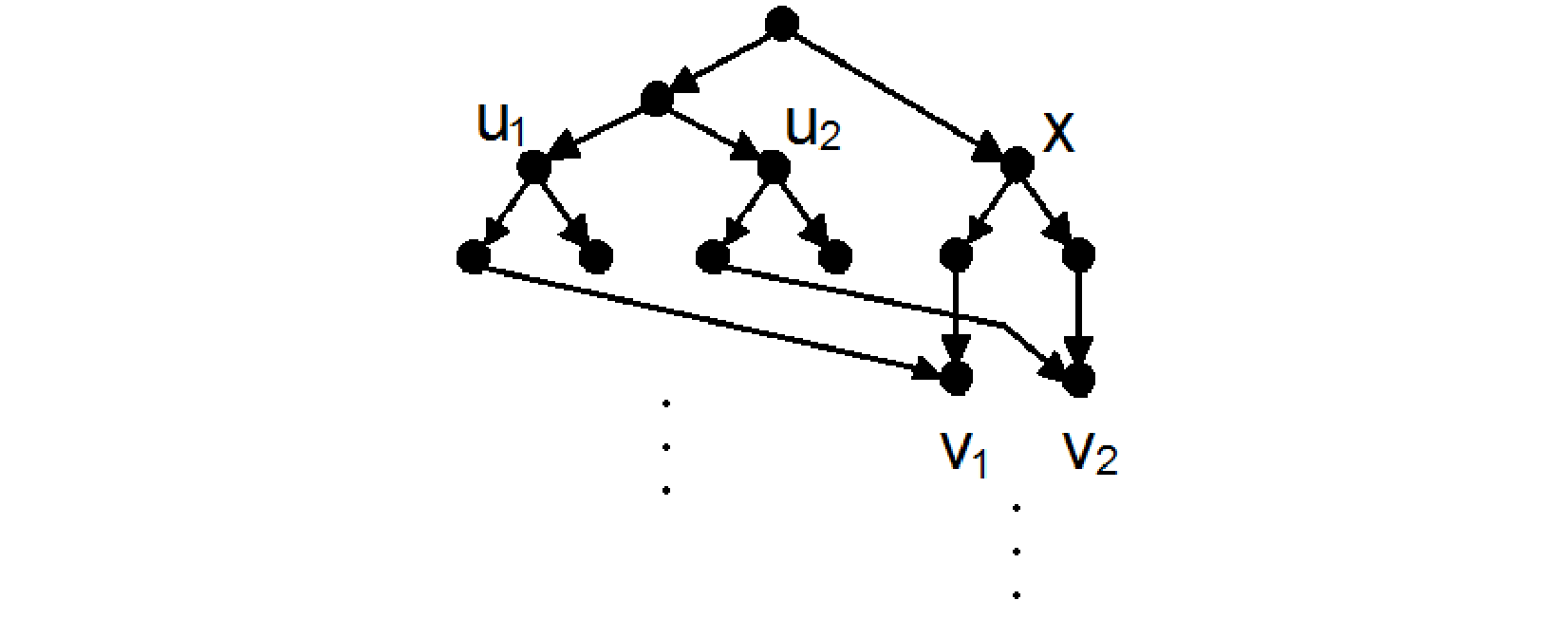}
\caption{A simplified version of the DAG in \cite{Spoonhower09}
that can incur high cache overhead.
Here, $v_1$ and $v_2$ are touches.
Suppose a processor $p_1$ executes the root node,
pushes the right child $x$ of the root node into its deque,
and then falls asleep.
Now another processor $p_2$ steals $x$ from $p_1$'s deque and executes the subgraph rooted at $x$.
Thus, $v_1$ and $v_2$ will be checked (to see if they are available) even before the corresponding future
threads are spawned at $u_1$ and $u_2$.}
\label{fig_unstructured}
\end{figure}

\begin{figure}
\centering
\includegraphics[width=7.5cm]{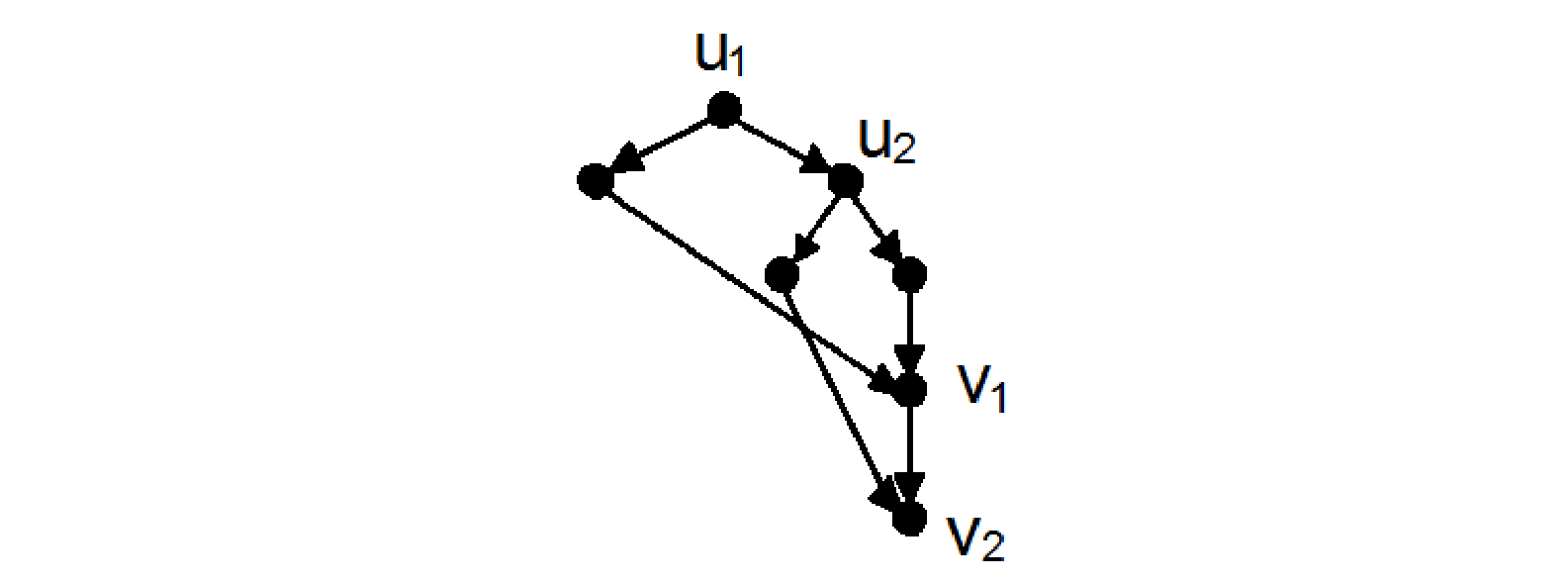}
\caption{In this structured (single-touch) computation, the touches $v_1$ and $v_2$
will not be checked until their corresponding future threads have been spawned
at $u_1$ and $u_2$, respectively.}
\label{fig_structured}
\end{figure}

\begin{definition}
A DAG is a \emph{structured future-parallel computation} if,
(1) for the future thread $t$ of any fork $v$,
the local parents of the touches of $t$ are descendants of $v$,
and (2) at least one touch of $t$ is a descendant of the right child of $v$.
\end{definition}
There are two reasons we require that at least one touch of $t$ is a descendant of the
right child of $v$.
First, it is natural that a computation spawns a future thread to compute a
future because the computation itself later needs that value.
At the fork $v$, the parent thread (the right child of $v$)
represents the ``main body" of the computation.
Hence, the future will usually be touched either by the parent thread,
or by threads spawned directly or indirectly by the parent thread.

Second, a computation usually needs a kind of ``barrier'' synchronization to
deal with resource release at the end of the computation.
Some node in the future thread $t$, usually the last node,
should have an outgoing edge pointing to the ``main body" of the
computation to tell the main body that the future thread has finished.
Without such synchronization,
$t$ and its descendants will be isolated from the main body
of the computation,
and we can imagine a dangerous scenario where the main body
of the computation finishes and releases its resources while
$t$ or its descendant threads are still running.

In our DAG model,
such a synchronization point is by definition a touch node,
though it may not be a real touch.
We follow the convention that
the thread that spawns a future thread releases it,
so the synchronization point is a node in the parent thread or one of its
descendants.
Another possibility is to place the synchronization point at the last node of
the entire computation,
which is the typically case in languages such as Java,
where the main thread of a program is in charge of releasing resources for the
entire  computation.
These two styles are essentially equivalent,
and should have almost the same bounds on cache overheads.
We will briefly discuss this issue in Section~\ref{section_superfinalnode}.

We consider how the following constraint affects cache locality.
\begin{definition}
A structured \emph{single-touch} computation is a structured future-parallel
computation where each future thread spawned at a fork $v$ is touched only once,
and the touch node is a descendant of $v$'s right child.
\end{definition}
By the definition of threads, the future parent of the only touch of a future thread must be the last node of the future thread (the last node can also be a parent of a join node, but we don't distinguish between a touch node and a join node).
The DAG in Figure \ref{fig_structured} represents a structured single-touch computation.
We will show that work-stealing parallel executions of structured single-touch
computations achieve significantly less cache overheads than unstructured computations.

In principle, a future could be touched multiple times by different threads,
so structured single-touch computations are more restrictive structured computations in general.
Nevertheless, the single-touch constraint is one that is likely to be observed by many programs.
For example, as noted,
the Cilk~\cite{Blumofe95} language supports fork-join parallelism,
a strict subset of the future-parallelism model considered here.
If we interpret the Cilk language's \texttt{spawn}
statement as creating a future,
and its \texttt{sync} statement as touching all untouched futures previously
created by that thread,
then Cilk programs (like all fork-join programs)
are structured single-touch computations.

%%%%%%%%%%%%%%%%%%%%%%%%%%%%%%%%%%%%%%%%%%%%
Structured single-touch computations encompass fork-join computations,
but are strictly more flexible. Figure~\ref{figure:ExamplesOfSingleTouch} presents two examples that illustrate the differences. If a thread creates multiple futures first and touches them later, fork-join parallelism requires they be touched (evaluated) in the reverse order. MethodA in Figure~\ref{figure:ExamplesOfSingleTouch}(a) shows the only order in which a thread can first create two futures and then touch them in a fork-join computation. This rules out, for instance, a program where a thread creates a sequence of futures, stores them in a priority queue, and evaluates them in some priority order. In contrast, our structured computations permit such futures to be evaluated by their creating thread or its descendants in any order.

Also, unlike fork-join parallelism, our notion of structured computation permits a thread to pass a future to another thread which touches that future, as illustrated in Figure~\ref{figure:ExamplesOfSingleTouch}(b): after a future is created, the future can be passed, as an argument of a new method call or the return value of the current thread's method call, to another thread.\footnote{
In previous versions of this paper, we misinterpreted the definition of structured single-touch computations
by stating that a future can only be passed as an argument of a method call,
but not as the return value of a method call.
In fact, the definition only implies a dependency path from the right child of the fork of a future,
which represents the parent thread that has created the future, to the touch node,
and therefore it does not rule out the possibility that a future can be pass as a return value
from one thread to another. }
The thread receiving the future (MethodC in the figure) can even pass it to another thread, and so on.
The only constraint is that only one of the threads that have received the future can touch it.
In a fork-join computation, however, only the thread creating the future can touch it, which is much more restrictive. We believe these restrictions are easy to follow and should be compatible with how many people program in practice.

\begin{figure}
\centering
\begin{tabbing}
\hspace{0.7in} void MethodA \{ \\
\hspace{0.9in}\= Future x = some computation;\\
\> Future y = some computation;\\
\> a = y.touch();\\
\> b = x.touch();\\
\hspace{0.7in} \}
\end{tabbing}
\subfigure[]{}

\begin{tabbing}
\\
\hspace{0.7in} void MethodB \{ \\
\hspace{0.9in}\= Future x = some computation;\\
\> Future y = MethodC(x);\\
\> ...... \\
\hspace{0.7in} \} \\

\hspace{0.7in} void MethodC(Future f)\{ \\
\hspace{0.9in} \= a = f.touch();\\
\> ...... \\
\hspace{0.7in} \}
\end{tabbing}
\subfigure[]{}
\caption{Two examples illustrating single-touch computations are more flexible than fork-join computations}
\label{figure:ExamplesOfSingleTouch}
\end{figure}

%%%%%%%%%%%%%%%%%%%%%%%%%%%%%%%%%%%%%%%%%%%%

\cite{Blelloch97} observe that if a future can be touched multiple times,
then complex and potentially inefficient operations and data structures are
needed to correctly resume the suspended threads that are waiting for the touch.
By contrast, the run-time support for futures can be significantly simplified
if each future is touched at most once.

We also consider the following structured local-touch computations in the paper.
\begin{definition}
A structured \emph{local-touch} computation is one where each future thread
spawned at a fork $v$ is touched only at nodes in its parent thread,
and these touches are descendants of the right child of $v$.
\end{definition}
Informally, the local touch constraint implies that a thread that needs the
value of a future should create the future itself.
Note that in a structured computation with local touch constraint,
a future thread is now allowed to evaluate multiple futures and these futures
can be touched at different times.
Though allowing a future thread to compute multiple futures is not very common,
\cite{Blelloch97} point out that it can be useful for some
future-parallel computations like pipeline
parallelism~\cite{Blelloch97,Blumofe95,Gordon06,Giacomoni08,Lee13}. We will show in
Section~\ref{section_localtouch} that work-stealing parallel executions of
computations satisfying the local touch constraint also have relatively low cache
overheads.
Note that structured computations with both single touch and local
touch constraints are still a superset of fork-join computations.

\section{Structured Single-Touch Computations}
\label{section:Structured Single-Touch Computations}
\subsection{Future Thread First at Each Fork}
\label{section_futurethreadfirst}
We now analyze cache performance of work stealing on parallel executions of
structured single-touch computations.
We will show that work stealing has relatively low cache overhead if the
processor at a fork always chooses the future thread to execute first,
and puts the parent future into its deque.
For brevity,
all the arguments and results in this section assume that
every execution chooses the future thread at a fork to execute first.

\begin{lemma}\label{lemmafutureparentfirst}
In the sequential execution of a structured single-touch computation, any touch $x$'s future parent is executed before $x$'s local parent, and the right child of $x$'s corresponding fork $v$ immediately follows $x$'s future parent.
\end{lemma}

\begin{proof}
By induction. Given a DAG,
initially let $S$ be an empty set and $T$ the set of all touches.
Note that
\begin{equation}
S \cap T = \emptyset \text{ and } S \cup T =\{\text{all touches}\}.
\label{eq:ST}
\end{equation}
Consider any touch $x$ in $T$,
such that $x$ has no ancestors in $T$.
(That is, $x$ has no ancestor nodes that are also touches.)
Let $t$ be the future thread of $x$ and $v$ the corresponding fork. Note that $x$'s future parent is the last node of $t$ by definition.
When the single processor executes $v$,
the processor pushes $v$'s right child into the deque and continues to execute thread $t$.
By hypothesis, there are no touches by $t$,
since any touch by $t$ must be an ancestor of $x$.
There may be some forks in $t$.
However, whenever the single processor executes a fork in $t$,
it pushes the right child of that fork,
which is a node in $t$, into the deque and
hence $t$ (i.e., a node in $t$) is right below $v$'s right child in the deque.
Therefore, the processor will always resume thread $t$ before the right child of $v$.
Since there is no touch by $t$,
all the nodes in $t$ are ready to execute one by one.
Thus, when the future parent of the touch $x$ is executed eventually,
the right child of $v$ is right at the bottom of the deque.
By the single touch constraint,
the local parent of $x$ is a  descendant of the right child of $v$,
so the local parent of $x$ cannot be executed yet.
Thus, the processor will now pop the right child of $v$ out
from the bottom of the deque.
Since this node is not a touch, it is ready to execute.
Therefore, $x$ satisfies the following two properties.
\begin{property}
Its future parent is executed before its local parent.
\label{prop:1}
\end{property}
\begin{property}
The right child of its corresponding fork immediately
follows its future parent.
\label{prop:2}
\end{property}
Now set $S = S \cup \{x\}$ and $T = T - \{x\}$.
Thus, all touches in $S$ satisfy Properties~\ref{prop:1} and~\ref{prop:2}.
Note that Equation~\ref{eq:ST} still holds.

Now suppose that at some point all nodes in $S$ satisfy Properties~\ref{prop:1}
and~\ref{prop:2}, and that Equation~\ref{eq:ST} holds.
Again, we now consider a touch $x$ in $T$, such that no touches in $T$ are
ancestors of $x$, i.e., all the touches that are ancestors of $x$ are in
$S$.
Since the computation graph is a DAG,
there must be such an $x$ as long as $T$ is not empty.
Let $t$ be the future thread of $x$ and $v$ the corresponding fork.
If there are no touches by $t$,
then we can prove $x$ satisfies Properties~\ref{prop:1} and~\ref{prop:2},
by the same argument for the first touch added into $S$.
Now assume there are touches by $t$.
Since those touches are ancestors of $x$,
they are all in $S$ and hence they all satisfy Property~\ref{prop:1}.
When the processor executes $v$,
it pushes $v$'s right child into the deque and starts executing $t$.
Similar to what we showed above,
when the processor gets to a fork in $t$,
it will always push $t$ into its deque,
right below the right child of $v$.
Thus, the processor will always resume $t$ before the right child of $v$.
When the processor gets to the local parent of a touch by $t$,
we know the future parent of the touch has already been executed since the
touch satisfies  Property~\ref{prop:1}.
Thus, the processor can immediately execute that touch and continue to execute $t$.
Therefore, the processor will eventually execute the future parent of $x$ while the right
child of $t$ is still the next node to pop in the deque.
Again, since the local parent of $x$ is a descendant of the right child of $v$,
the local parent of $x$ as well as $x$ cannot be executed yet.
Therefore, the processor will now pop the right child of $v$ to execute,
and hence $x$ satisfies Properties~\ref{prop:1} and~\ref{prop:2}.
Now we set $S = S \cup \{x\}$ and $T = T - \{x\}$.
Therefore, all touches in $S$ satisfy Properties~\ref{prop:1} and~\ref{prop:2},
and Equation\ref{eq:ST}) also holds.
By induction,
we have $S= \{\text{all touches}\}$ and all touches satisfy
Properties~\ref{prop:1} and~\ref{prop:2}.
\end{proof}

\cite{Acar00} have shown that the number of additional cache misses in a
work-stealing parallel computation is bounded by the product of the number of
deviations and the number of cache lines.
It is easy to see that only two types of nodes in a DAG can be deviations:
the touches and the child nodes of forks that are not chosen to execute
first.
Since we assume the future thread (left child) at a fork is always executed first,
only the right children of forks can be deviations.
Next, we bound the number of deviations incurred by a work-stealing parallel
execution to bound its cache overhead.

\begin{lemma}\label{deviationtouch}
Let $t$ be the future thread at a fork $v$ in a structured single-touch computation.
If $t$'s touch $x$ or $v$'s right child $u$ is a deviation,
then either $u$ is stolen or there is a touch by $t$ which is a deviation.
\end{lemma}

\begin{proof}
By Lemma~\ref{lemmafutureparentfirst},
a touch is a deviation if and only if its local parent is executed before its future parent.
Now suppose a processor $p$ executes $v$ and pushes $u$ into its deque.
Assume that $u$ is not stolen and no touches by $t$ are deviations.
Thus, $u$ will stay in $p$'s deque until $p$ pops it out.
The proof of this lemma is similar to that of Lemma~\ref{lemmafutureparentfirst}.
After $p$ spawns thread $t$ at $v$, it moves to execute $t$.
When $p$ executes ``ordinary" nodes in $t$, no nodes are pushed into or popped
out of $p$'s deque and hence $u$ is still the next node in the deque to pop.
When $p$ executes a fork in $t$, it pushes $t$ (more specifically,
the right child of that fork) into its deque, right below $u$.
Since a thief processor always steals from the top of a deque,
and by hypothesis $u$ is not stolen, $t$ cannot be stolen.
Thus, $p$ will always resume $t$ before $u$
and then $u$ will become the next node in the deque to pop.
When $p$ executes the local parent of a touch by $t$,
the future parent of that touch must have been executed,
since we assume that touch is not a deviation.
Thus, $p$ can continue to execute that touch immediately and keep moving on in $t$
with its deque unchanged.
Therefore, $p$ will finally get to the local parent of $x$ and then pop $u$ out from
its deque, since $x$ is a descendant of $u$ and $x$ cannot be execute yet.
Hence, neither $x$ nor $u$ can be a deviation.
\end{proof}

\begin{theorem}\label{futurefirstupperbound}
If, at each fork, the future thread is chosen to execute first,
then a parallel execution with work stealing incurs $O(PT_{\infty}^2)$ deviations and
$O(CPT_{\infty}^2)$ additional cache misses in expectation on a structured single-touch
computation,
where (as usual) $P$ is the number of processors involved in this computation,
$T_{\infty}$ is the computation span, and $C$ is the number of cache lines.
\end{theorem}
\begin{proof}
\cite{Arora98} have shown that in a parallel execution with work stealing,
there are in expectation $O(PT_{\infty})$ steals.
Now let us count how many deviations these steals can incur.
A steal on the right child $u$ of a fork $v$ can make $u$ and $v$'s
corresponding touch $x_1$ deviations.
Suppose $x_1$ is a touch by a thread $t_2$,
then the right child of the fork of $t_2$ and $t_2$'s touch $x_2$ can be
deviations.
If $x_2$ is a deviation and $x_2$ is a touch by another thread $t_3$,
then the right child of the fork of $t_3$ and $t_3$'s touch $x_3$
can be deviation too.
Note that $x_2$ is a descendant of $x_1$ and $x_3$ is a descendant of $x_2$.
By repeating this observation,
we can find a chain of touches $x_1, x_2, x_3,..., x_n$,
called a \emph{deviation chain},
such that each $x_i$ and the right child of the corresponding fork of
$x_i$ can be deviations.
Since for each $i > 1$, $x_i$ is a descendant of $x_2$,
$x_1, x_2, x_3, \ldots, x_n$ is in a directed path in the computation DAG.
Since the length of any path is at most $T_{\infty}$,
we have $n \le T_{\infty}$.
Since each future thread has only one touch, there is only one deviation chain for a steal.
Since there are $O(PT_{\infty})$ steals in expectation in a parallel execution~\cite{Arora98},
we can find in expectation $O(PT_{\infty})$ deviation chains and in total
$O(PT_{\infty}^2)$ touches and right children of the corresponding forks involved, i.e., $O(PT_{\infty}^2)$ deviations involved.

Next, we prove by contradiction that no other touches or right children of
forks can be deviations. suppose there is touch $y$, such that $y$ or
the right child of the corresponding fork of $y$ is a deviation, and that $y$
is not in any deviation chain. The right child of the corresponding fork of
$y$ can not be stolen, since by hypothesis $y$ is not the first touch in any of those
chains. Thus by Lemma~\ref{deviationtouch}, there is a touch $y'$ by the future
thread of $y$ and $y'$ is a deviation. Note that $y's$ cannot be in any
deviation chain either. Otherwise $y$ and the deviation chain $y'$ is in will
form a deviation chain too, a contradiction. Therefore, by repeating such ``tracing back", we
will end up at a deviation touch that is not in any deviation chain and
has no touches as its ancestors.
Therefore, there are no touches by the
future thread of this touch, and the right child of the corresponding future
fork of it is not stolen, contradicting Lemma~\ref{deviationtouch}.

The upper bound on the expected number of additional cache misses follows from the result of \cite{Acar00} that the number of additional cache misses in a
work-stealing parallel computation is bounded by the product of the number of
deviations and the number of cache lines.
\end{proof}

The bound on the number of deviations in Theorem~\ref{futurefirstupperbound} is
tight, and the bound on the number of additional cache misses is tight within a
factor of $C$, as shown below in Theorem~\ref{futurefirstlowerbound}.

\begin{figure}
\centering

\subfigure[]{\includegraphics[width=3.2cm]{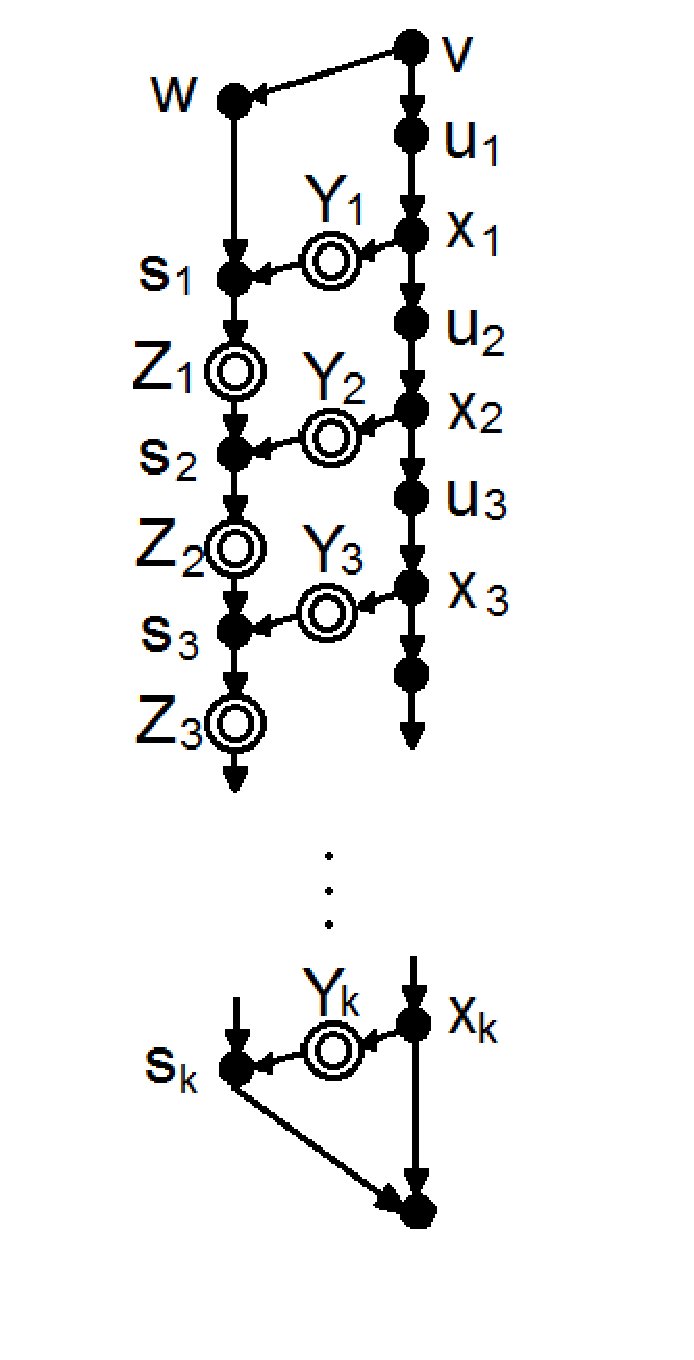}}
\subfigure[]{\includegraphics[width=4.0cm]{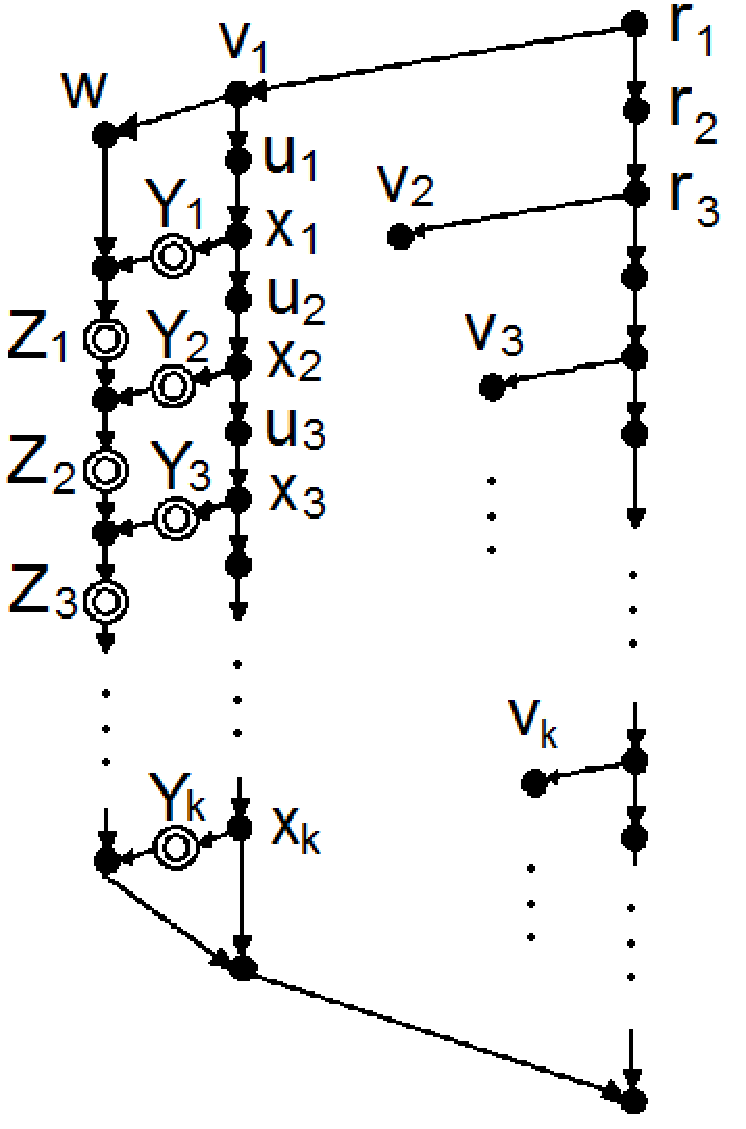}}
\subfigure[]{\includegraphics[width=3.2cm]{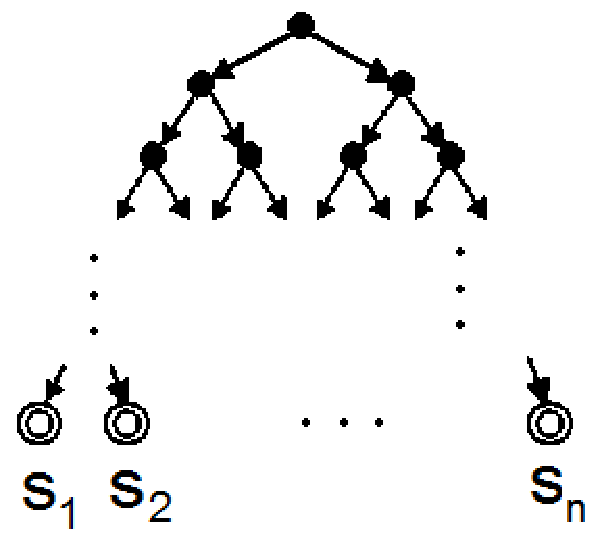}}

\caption{Figure (c) shows a DAG on which work stealing can incur $\Omega(PT_{\infty}^2)$ deviations and $\Omega(PT_{\infty}^2)$ additional cache misses. It uses the DAGs in (a) and (b) as building blocks.}
\label{fig_futurefirst_lowerbound}
\end{figure}

\begin{theorem}\label{futurefirstlowerbound}
If, at each fork node, the future thread is chosen to execute first, then
a parallel execution with work stealing can incur $\Omega(PT_{\infty}^2)$
deviations and $\Omega(PT_{\infty}^2)$ additional cache misses on a structured
single-touch computation,
while the sequential execution of this computation incurs $O(PT_{\infty}^2/C)$ cache misses.
\end{theorem}

\begin{proof}
Figure~\ref{fig_futurefirst_lowerbound}(c) shows a computation DAG on which we can get the bounds we want to prove. The DAG in Figure~\ref{fig_futurefirst_lowerbound}(c) uses the DAGs in Figures~\ref{fig_futurefirst_lowerbound}(a) and \ref{fig_futurefirst_lowerbound}(b) as building blocks. Let's look at Figures~\ref{fig_futurefirst_lowerbound}(a) first. Suppose there are two processors $p_1$ and $p_2$ executing the DAG in Figure~\ref{fig_futurefirst_lowerbound}(a). Suppose $p_2$ executes $v$, pushes $u_1$ into its deque, and then falls asleep before executing $w$. Now suppose $p_1$ steals $u_1$. For each $i \le k$, neither $s_i$ nor $Z_i$ can be
executed since $w$ has not been executed yet. Now $p_1$ takes a solo run, executing $u_1, x_1, Y_1, u_2, x_2, Y_2,..., x_k, Y_k$. After $p_1$ finishes, $p_2$ wakes up and executes the rest of the computation DAG. Note that the right (local) parent of $s_i$ is executed before the left (future) parent of the touch is executed. Thus, by Lemma~\ref{lemmafutureparentfirst}, each $s_i$ is a deviation. Hence, this parallel execution incurs $k$ deviations and the computation span of the computation is $\Theta(k)$.

Now let us consider a parallel execution of the computation in~\ref{fig_futurefirst_lowerbound}(b). For each $i\le k$, the subgraph rooted at $v_i$ is identical to the computation DAG in~\ref{fig_futurefirst_lowerbound}(a) (except that the last node of the subgraph has an extra edge pointing to a node of the main thread). Suppose there are three processors $p_1$, $p_2$, and $p_3$ working on the computation. Assume $p_2$ executes $r_1$ and $v_1$ and then falls asleep when it is about to execute $w$. $p_3$ now steals $r_2$ from $p_2$ and then falls asleep too. Then $p_1$ steals $u_1$ from $p_2$'s deque. Now $p_1$ and $p_2$ execute the subgraph rooted at $v_1$ in the same way they execute the DAG in~\ref{fig_futurefirst_lowerbound}(a). After $p_1$ and $p_2$ finish, $p_3$ wakes up, executes $r_2$. Now these three processors start working on the subgraph rooted at $r_3$ in the same way they executed the graph rooted at $r_1$. By repeating this, the execution ends up incurring $k^2$ deviations when all the $k$ subgraphs are done. Since the length of the path $r_1,r_2,r_3...$ on the right-hand side is $\Theta(k)$, the computation span of the DAG is still $\Theta(k)$.

Now we construct the final computation DAG, as in Figure~\ref{fig_futurefirst_lowerbound}(c). The ``top" nodes of the DAG are all forks, each spawning a future thread. Thus, they form a binary tree and the number of threads increase exponentially. The DAG stops creating new threads at level $\Theta(\log n)$ when it has $n$ threads rooted at $S_1, S_2,..., S_n$, respectively. For each $i$, the subgraph rooted at $S_i$ is identical to the DAG in~\ref{fig_futurefirst_lowerbound}(b). Suppose there are $3n$ processors working on the computation. It is easy to see $n$ processors can eventually get to $S_1, S_2,..., S_n$. Suppose they all fall asleep immediately after executing the first two nodes of $S_i$(corresponding to $r_1$ and $v_1$ in Figure~\ref{fig_futurefirst_lowerbound}(b)) and then each two of the rest $2n$ free processors join to work on the subgraph rooted at $S_i$, in the same way $p_1$, $p_2$ and $p_3$ did in Figure~\ref{fig_futurefirst_lowerbound}(b). Therefore, this execution will finally incur $nk^2$ deviations, while the computation span of the DAG is $\Theta(k+\log n)$. Therefore, by setting $n=P/3$, we get a parallel execution that incurs $\Omega(PT_{\infty}^2)$ deviations, when $\log P = O(k)$.

To get the bound on the number of additional cache misses, we just need to
modify the graph in~\ref{fig_futurefirst_lowerbound}(a) as follows. For each
$1\le i \le k$, $Y_i$ consists of a chain of $C$ nodes $y_{i1}, y_{i2},...,
y_{iC}$, where $C$ is the number of cache lines. $y_{i1}, y_{i2},..., y_{iC}$
access memory blocks $m_1, m_2,...,m_C$, respectively. Similarly, each $Z_i$
consists of a chain of $C$ nodes $z_{i1}, z_{i2},..., z_{iC}$. $z_{i1},
z_{i2},..., z_{iC}$ access memory blocks $m_C, m_{C-1},..., m_1$,
respectively. all $s_i$ access memory block $m_C$. For all $1\le i \le k$,
$u_i$ and $x_i$ both access memory block $m_{C+1}$. It does not matter which
memory blocks the other nodes in the DAG access. For simplicity, assume the
other nodes do not access memory. In the sequential execution, the single
processor has $m_1, m_2,...,m_C$ in its cache after executing
$v,w,u_1,x_1,Y_1,Z_1$ and it has incurred $(C+1)$ cache misses so far. Now it
executes $u_2$ and $x_2$, incurring one cache miss at node $u_2$ by replacing
$m_{C}$ with $m_{C+1}$ in its cache, since $m_{C}$ is the least recently used
block. When it executes $Y_2$ and $Z_2$, it only incurs one cache miss by
replacing $m_{C+1}$ with $m_{C}$ at the last node of $Y_2$, $y_{2C}$. Likewise,
it is easy to see that the sequential execution will only incur cache misses at
nodes $u_i$ and at the last nodes of $Y_i$ for all $i$. Hence, the sequential
execution incurs only $O(k+C)$ cache misses. When $k = \Omega(C)$, the
sequential execution incurs only $O(k)$ cache misses.

Now consider the parallel
execution by two processors $p_1$ and $p_2$ we described before. $p_2$ will incur only $C$ cache
misses, since $Z_i$ and $s_i$ only access $m$ different blocks $m_1,
m_2,...,m_C$ and hence $p_2$ doesn't need to swap any memory blocks out of its cache. However, $p_1$ will incur lots of cache misses. After executing
each $Y_i$, $p_1$ will execute $u_{i+1}$. Thus at $u_{i+1}$, one cache miss is
incurred and $m_{1}$ is replaced with $m_{C+1}$, since $m_{1}$ is the least
recently used block. Then, when $p_1$ executes the first node $y_{(i+1)1}$ in $Y_i$,
, $m_1$ is not in its cache. Since $m_{2}$ now becomes the least
recently used memory block in $p_1$'s cache, $m_2$ is replaced by
$m_1$. Thus, $m_2$ will not be in the cache when it is in need at $y_{(i+1)2}$. Therefore, it is obvious that $p_1$ will incur a cache miss at each node in $Y_i$ and hence incur $Ck$ cache misses in total in the entire execution. Note that the computation span of this modified DAG is $\Theta(Ck)$, since each $Z_i$ now has $C$ nodes. Therefore, the sequential execution and the parallel execution actually incur $\Theta(T_{\infty}/C)$ and $\Theta(T_{\infty})$, respectively, when $\log P = O(k)$. Therefore, if we use this modified DAG as the building blocks in~\ref{fig_futurefirst_lowerbound}(c), we will get the bound on the number of additional cache misses stated in the theorem.
\end{proof}

\subsection{Parent Thread First at Each Fork}
\label{section:Parent Thread First at Each Fork}
In this section, we show that if the parent thread is always executed first at a fork, a work-stealing parallel execution of a structured single-touch computation can incur $\Omega(tT_{\infty})$ deviations and $\Omega(CtT_{\infty})$ additional cache misses, where $t$ is the number of touches in the computation, while the corresponding sequential execution incurs only a small number of cache misses.
This bound matches the upper bound for general, unstructured future-parallel computations \cite{Spoonhower09}\footnote[2]{The bound on the expected number of deviations in  \cite{Spoonhower09} is actually $O(PT_{\infty}+tT_{\infty})$. However, as pointed out in  \cite{Spoonhower09}, a simple fork-join computation can get $\Omega(PT_{\infty})$ deviations. Hence we focus on the more interesting part $\Omega(tT_{\infty})$.}. This result, combined with the result in Section~\ref{section_futurethreadfirst}, shows that choosing the future threads at forks to execute first achieves better cache locality for work-stealing schedulers on structured single-touch computations.

\begin{figure}
\centering

\subfigure[]{\includegraphics[width=4.8cm]{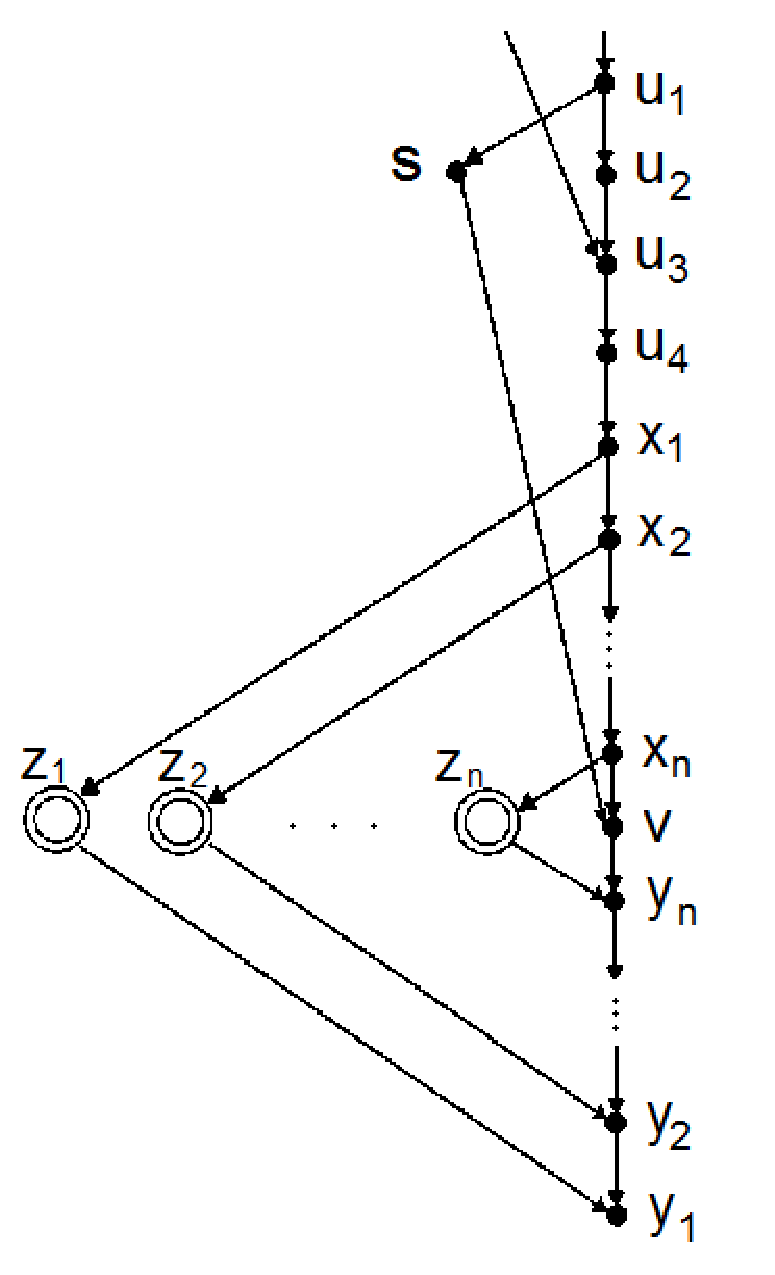}}
\subfigure[]{\includegraphics[width=4.8cm]{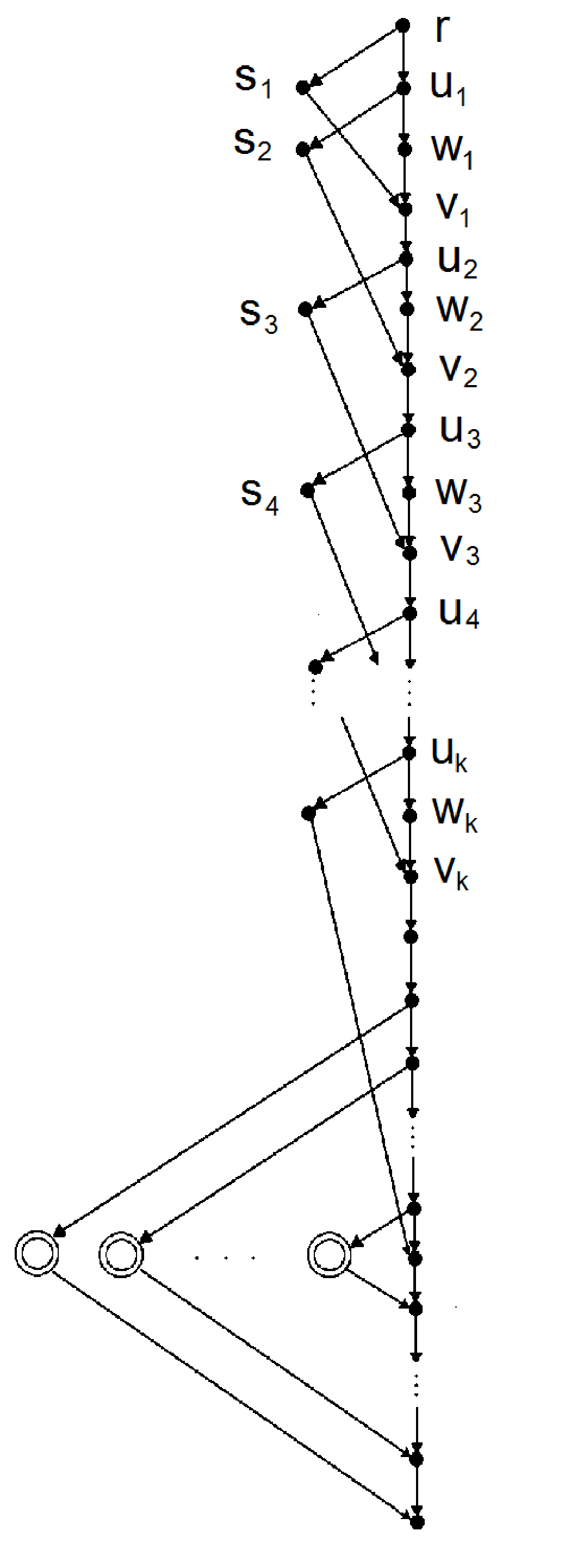}}

\caption{DAGs used by Figure~\ref{fig_parentfirst_lowerbound2} as building blocks.}
\label{fig_parentfirst_lowerbound1}
\end{figure}

\begin{theorem}\label{parentthreadfirst}
If, at each fork, the parent thread is chosen to execute first, then a parallel execution with work stealing can incur $\Omega(tT_{\infty})$ deviations and $\Omega(CtT_{\infty})$ additional cache misses on a structured single-touch computation, while the sequential execution of this computation incurs only $O(C+t)$ cache misses.
\end{theorem}

\begin{proof}
The final DAG we want to construct is in Figure~\ref{fig_parentfirst_lowerbound2}. It uses the DAGs in Figure~\ref{fig_parentfirst_lowerbound1} as building blocks.
We first describe how a single deviation at a touch $u_3$ can incur
$\Omega(T_{\infty})$ deviations and $\Omega(CT_{\infty})$ additional cache
misses in Figure~\ref{fig_parentfirst_lowerbound1}(a). In order to get the
bound we want to prove, here we follow the convention in
\cite{Acar00,Spoonhower09} to distinguish between touches and join nodes in the
DAG. More specifically, $y_i$ is a join node, not a touch, for each $1\le i\le
n$. For each $1\le i \le n$, node $x_i$ accesses memory block $m_1$ and $y_i$
accesses memory block $m_{C+1}$. $Z_i$ consists of a chain of $C$ nodes
$z_{i1},z_{i2},...,z_{iC}$, accessing memory blocks $m_1,m_2,...,m_{C}$
respectively. All the other nodes do not access memory. Assume in the sequential execution a single processor $p_1$ executes the entire DAG in Figure~\ref{fig_parentfirst_lowerbound1}(a). Suppose initially the left (future) parent of $u_3$ has already been executed. $p_1$ starts executing the DAG at $u_1$. Since $p_1$ always stays on the parent thread at a fork, it first pushes $s$ into its deque, continues to execute $u_2,u_3,u_4$, and then executes $x_1,x_2,...,x_n$ while pushing $z_{11}, z_{21},..., z_{n1}$ into its deque. Since $v$ cannot be executed due to $s$, $p_1$ pops $z_{n1}$ out of its deque and executes the nodes in $Z_n$. Then $p_1$ executes all the nodes in $Z_{n-1}, Z_{n-2},...,Z_{1}$, in this order. So far $p_1$ has only incurred $C$ cache misses, since all the nodes it has executed only access memory blocks $m_1,...,m_{C}$ and hence it did not need to swap any memory blocks out of its cache. Now $p_1$ executes $s, v$ and then $y_n,y_{n-1},...,y_1$, incurring only one more cache miss by replacing $m_1$ with $m_{C+1}$ at $y_n$. Hence, this execution incurs $O(C)$ cache misses in total. Note that the left parent of $y_i$ is executed before the right parent $y_i$ for all $i$.

Now assume in another execution by $p_1$, the left parent of $u_3$ is in $p_1$'s deque when $p_1$ starts executing $u_1$. Thus, $u_3$ is a deviation with respect to the previous execution. Since $u_3$ is not ready to execute after $p_1$ executes $u_2$, $p_1$ pops $s$ out of its deque to execute. Since $v$ is not ready, $p_1$ now pops the left parent of $u_3$ to execute and then executes $u_3, u_4, x_1, x_2,...,x_n, v$. Now $p_1$ pops $z_{n1}$ out and executes all the nodes $Z_n$. Note that $y_n$ is now ready to execute and the memory blocks in $p_1$'s cache at the moment are $m_1, m_2,..., m_C$. Now $p_1$ executes $y_n$, replacing the least recently used block $m_1$ with $m_{C+1}$. $p_1$ then pops $z_{(n-1)1}$ out and executes all the nodes $z_{(n-1)1}, z_{(n-1)2},...,z_{(n-1)}C$ in $Z_{n-1}$ one by one. When $p_1$ executes $z_{(n-1)1}$, it replaces $m_2$ with $m_1$, and when it executes $z_{(n-1)2}$, it replaces $m_3$ with $m_2$, and so on. The same thing happens to all $Z_i$ and $y_i$. Thus, $p_1$ will incur a cache miss at every node afterwards, ending up with $\Omega(Cn)$ cache misses in total. Note that the computation span of this DAG is $T_{\infty} = \Theta(C+n)$. Thus, this execution with a deviation at $u_3$ incurs $\Omega(CT_{\infty})$ cache misses when $n = \Omega(C)$. Moreover, all $y_i$ are deviations and hence this execution incurs $\Omega(T_{\infty})$ deviations.

Now let us see how a single steal at the beginning of a thread results in $\Omega(T_{\infty})$ deviations and $\Omega(CT_{\infty})$ cache misses at the end of the thread. Figure~\ref{fig_parentfirst_lowerbound1}(b) presents such a computation. First we consider the sequential execution by a processor $p_1$. It is easy to see $p_1$ executes nodes in the order $r, u_1, w_1, s_2$, $s_1, v_1, u_2$, $w_2, v_2, u_3, w_3$, $s_4, s_3, v_3, u_4,...$. The key observation is that $w_i$ is executed before $s_i$ is executed for any odd-numbered $i$ while $w_i$ is executed after $s_i$ is executed for any even-numbered $i$. This statement can be proved by induction. Obviously, this holds for $i=1$ and $i=2$, as we showed before. Now suppose this fact holds for all $1,2,...,i$, for some even-numbered $i$. Now suppose $p_1$ executes $u_{i-1}$. Then $p_1$ pushes $s_i$ into its deque and executes $w_{i-1}$.
Since we know $w_{i-1}$ should be executed before $s_{i-1}$, $s_{i-1}$ has not been executed yet.
Moreover, $s_{i-1}$ must already be in the deque before $s_{i}$ was pushed into the deque, since $s_{i-1}$'s parent $u_{i-2}$ has been executed and $s_{i-1}$ is ready to execute.
Now $p_1$ pops $s_i$ out to execute.
Since $v_i$ is not ready to execute, $p_1$ pops $s_{i-1}$ out
and then executes $s_{i-1}, v_{i-1}, u_i$, and pushes $s_{i+1}$ into the deque.
Now $p_1$ continues to execute $w_i, v_i, u_{i+1}$ and pushes $s_{i+1}$ into its deque.
Then $p_i$ executes $w_{i+1}$ and pops $s_{i+2}$ out, since $v_{i+1}$
is not ready due to $s_{i+1}$.
Now we can see $w_{i+1}$ and $s_{i+2}$ have been executed, but $s_{i+1}$ and $w_{i+2}$ not yet. That is, $w_{i+1}$ is executed before $s_{i+1}$ and $w_{i+2}$ is executed after $s_{i+2}$.
Therefore, the statement holds for $i+1$ and $i+2$, and hence the proof completes.

The subgraph rooted at $u_k$ is identical to the graph in Figure~\ref{fig_parentfirst_lowerbound1}(a), with $v_k$ corresponding to $u_3$ in Figure~\ref{fig_parentfirst_lowerbound1}(a). Therefore, if $k$ is an even number, $v_k$'s left parent has been executed when $w_k$ is executed and hence the sequential execution will incur only $O(C)$ cache misses on the subgraph rooted at $u_k$.

Now consider the following parallel execution of the DAG in
Figure \ref{fig_parentfirst_lowerbound1}(b) by two processors $p_1$ and $p_2$.
$p_1$ first executes $r$ and pushes $s_1$ into its deque.
Then $p_2$ immediately steals $s_1$ and executes it.
Now $p_2$ falls asleep, leaving $p_1$ executing the rest of the DAG alone.
It is easy to see $p_1$ will execute the nodes in the DAG in the order
$u_1, w_1, v_1, u_2, w_2, s_3, s_2, v_2, u_3, w_3, v_3, u_4, s_4,...$
It can be proved by induction that $w_i$ is executed after $s_i$ is executed for
any odd-numbered $i$ while $w_i$ is executed before $s_i$ is executed for
any even-numbered $i$, which is opposite to the order in the sequential execution.
The induction proof is similar to that of the previous observation
in the sequential execution, so we omit the proof here.
If $k$ is an even number, $w_k$ will be executed before the left parent of $v_k$
and hence this execution will incur $\Omega(T_{\infty})$ deviations and
$\Omega(CT_{\infty})$ cache misses when $n = \Omega(C)$ and $n = \Omega(k)$.

The final DAG we want to construct is in Figure~\ref{fig_parentfirst_lowerbound2}.
This is actually a generalization of the DAG in Figure~\ref{fig_parentfirst_lowerbound1}(b).
Instead of having one fork $u_i$ before each touch $v_i$, it has two forks $u_i$ and $x_i$,
for each $i$. After each touch $v_i$, the thread at $y_i$ splits into two identical branches,
touching the futures spawned at $u_i$ and $x_i$, respectively.
In this figure, we only depict the right branch and omit the identical left branch.
As we can see, the right branch later has a touch $v_{i+1}$ touching the future $s_{i+1}$
spawned at the fork $x_i$. If we only look at the thread on the right-hand side,
it is essentially the same as the DAG in Figure\ref{fig_parentfirst_lowerbound1}(b).
The sequential execution of this DAG by $p_1$ is similar to that in Figure\ref{fig_parentfirst_lowerbound1}(b).
The only difference is that $p_1$ at each $y_i$ will execute the right branch
first and then the left branch recursively.
Similarly, it can be proved by induction that $w_i$ is executed before $s_i$ is executed for any odd-numbered $i$ while $w_i$ is executed after $s_i$ is executed for any even-numbered $i$. Obviously this also holds for each left branch. Now consider a parallel execution by two processors $p_1$ and $p_2$. $p_1$ first executes $r$. $p_2$ immediately steals $s_1$ and executes it and then sleeps forever. Now $p_1$ makes a solo run to execute the rest of the DAG. Again, we can prove by the same induction argument that $w_i$ is executed after $s_i$ is
executed for any odd-numbered $i$ while $w_i$ is executed before $s_i$ is executed
for any even-numbered $i$, which is opposite to the order in the sequential execution.
The above two induction proofs are a little more complicated than those for the DAG in Figure\ref{fig_parentfirst_lowerbound1}(b), but the ideas are essentially the same
(the only difference is now we have to prove the statements hold for the two identical
branches split at fork $y_i$ at the inductive step) and hence we omit the proofs again.

By splitting each thread into two after each $y_i$, the number of branches in the DAG increases exponentially. Suppose there are $t$ touches in the DAG. Thus, there are eventually $\Theta(t)$ branches and the height of this structure is $\Theta(\log t)$. At the end of each branch is a subgraph identical to the DAG in Figure~\ref{fig_parentfirst_lowerbound1}(a). Therefore, the parallel execution with only one steal can end up incurring $\Theta(tn)$ deviations and $\Theta(Ctn)$ cache misses. The sequential execution incurs only $\Theta(C+t)$ cache misses, since the sequential execution will incur only 2 cache misses by swapping $m_{C+1}$ in and out at each branch, after it incurs $C$ cache misses to load $m_1,m_2, ..., m_{C}$ at the first branch. hence, when $n = \Omega(\log t)$ and $n = \Omega(C)$, we get the bound stated in the theorem.
\end{proof}

%\begin{figure}
\begin{figure}
\centering

\includegraphics[width=8cm]{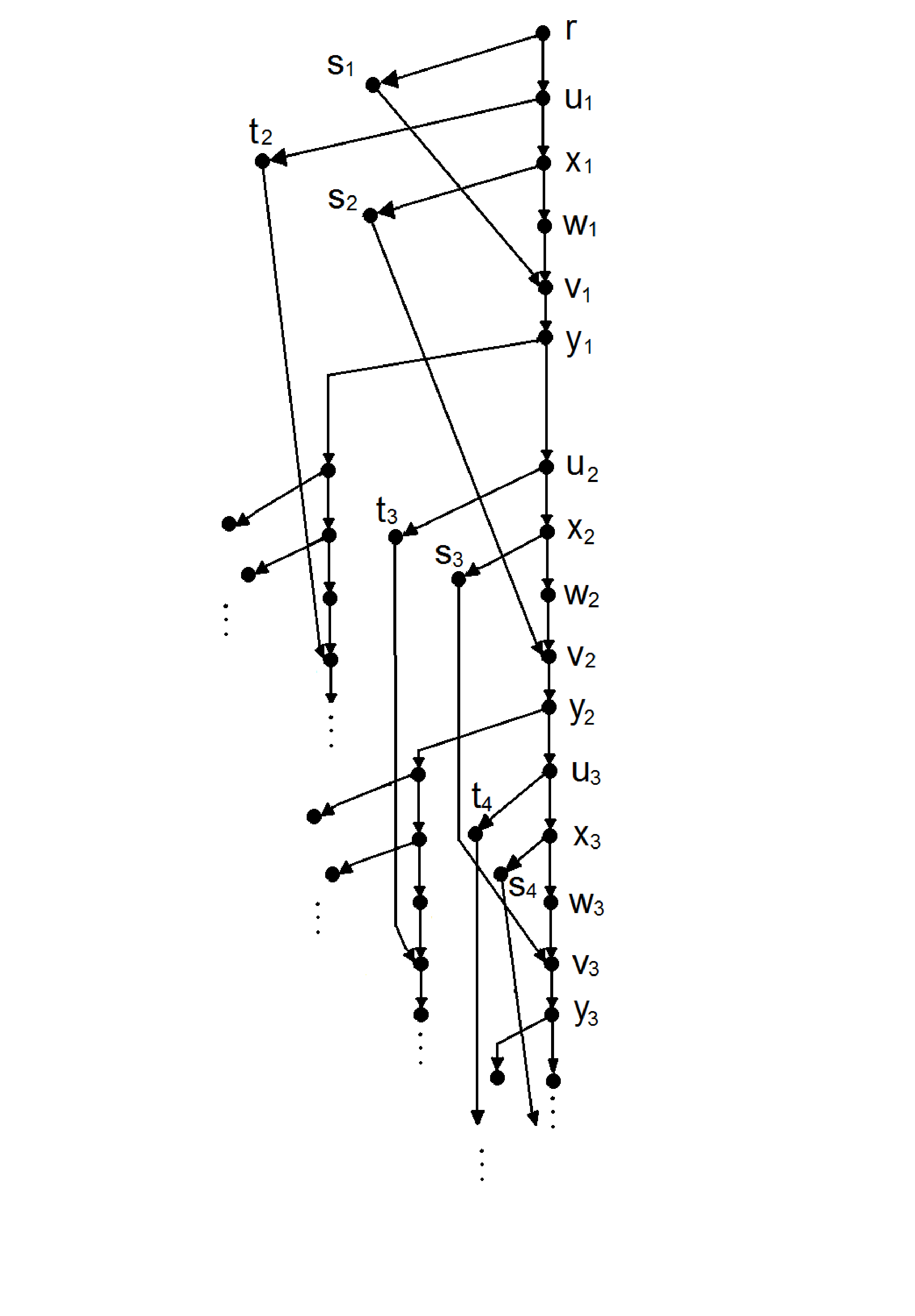}

\caption{A DAG on which work stealing can incur $\Omega(tT_{\infty})$
deviations and $\Omega(CtT_{\infty})$ if it chooses parents threads to execute
first at forks.
This example uses the DAGs in Figure~\ref{fig_parentfirst_lowerbound1} as building blocks.}
\label{fig_parentfirst_lowerbound2}
%\end{figure}
\end{figure}

\section{Other Kinds of Structured Computations}
\label{section:Other Kinds of Structured Computations}
It is natural to ask whether other kinds of structured computations can also
achieve relatively good cache locality.
We now consider two alternative kinds of restrictions.

\subsection{Structured Local-Touch Computations}\label{section_localtouch}
In this section, we prove that work-stealing parallel executions of
structured local-touch computations also have relatively good cache
locality,
if the future thread is chosen to execute first at each fork.
This result, combined with Theorems~\ref{futurefirstupperbound}
and~\ref{parentthreadfirst},
implies that work-stealing schedulers for structured computations
are likely better off choosing future threads to execute first at forks.

\begin{lemma}\label{lemmafutureparentfirstlocaltouch}
In the sequential execution of a structured local-touch computation where the future thread at a fork is always chosen to execute first,
any touch $x$'s future parent is executed before $x$'s local parent,
and the right child of any fork $v$ immediately follows the last node of
the future thread spawned at $v$, i.e.,
the future parent of the last touch of the future thread.
\end{lemma}

The proof is omitted because it is almost identical to that of Lemma
\ref{lemmafutureparentfirst}. (We first consider a future thread whose touches are the
``earliest" in the DAG, that is, no other touches are ancestors of them,
and we can easily prove
the statement in Lemma \ref{lemmafutureparentfirstlocaltouch} holds for those touches.
Then by the same induction proof as for Lemma \ref{lemmafutureparentfirst},
we can prove the statement holds for all future threads' touches.)

\begin{theorem}\label{localtouchupperbound}
If the future thread at a fork is always chosen to execute first,
then a parallel execution with work stealing incurs $O(PT_{\infty}^2)$ deviations and
$O(CPT_{\infty}^2)$ additional cache misses in expectation on a structured
local-touch computation.
\end{theorem}

\begin{proof}
Let $v$ be a fork that spawns a future thread $t$. Now we consider a parallel execution. Let $p$ be a processor that executes $v$ and pushes the right child of $v$ into its deque. Suppose the right child of $v$ is not stolen. Now consider the subgraph $G'$ consisting of $t$ and its descendant threads. Note that $G'$ itself is a structured computation DAG with local touch constraint. Now $p$ starts executing $G'$.

According to local touch constraint, the only nodes outside $G'$ that connect to the nodes in $G'$ are $v$ and the touches of $t$, and $c$ is the only node outside $G'$ that the nodes in $G'$ depend on. Now $v$ has been executed and the touches of $t$ are not ready to execute due to the right child of $v$. Hence, $p$ is able to make a sequential execution on $G'$ without waiting for any node outside to be done or jumping to a node outside, as long as no one steals a node in $G'$ from $p$'s deque. Since we assume the right child of $v$ will not be stolen and any nodes in $G'$ can only be pushed into $p$'s deque below $v$, no nodes in $G'$ can be stolen. Hence, $G'$ will be executed by a sequential execution by $p$. Therefore, there are no deviations in $G'$. After $p$ executed the last node in $G'$, which is the last node in $t$, $p$ pops the right child of $v$ to execute. Hence, the right child of $v$ cannot be a deviation either, if it is not stolen. That is, those nodes can be deviations only if the right child of $v$ is stolen. Since there are in expectation $O(PT_{\infty})$ steals in an parallel execution and each future thread has at most $T_{\infty}$ touches, the expected number of deviations is bounded by $O(PT_{\infty}^2)$ and the expected number of additional touches is bounded by $O(CPT_{\infty}^2)$.
\end{proof}

\subsection{Structured Computations with Super Final Nodes}\label{section_superfinalnode}
As discussed in Section~\ref{section_structured},
in languages such as Java,
the program's main thread typically releases all resources at the end of an execution.
To model this structure,
we add an edge from the last node of each thread to the final node of the computation DAG.
Thus, the final node becomes the only node with in-degree greater than 2.
Since the final node is always the last to execute,
simply adding those edges pointing to the final node into a DAG will not change
the execution order of the nodes in the DAG.
It is easy to see that having such a super node will not change the upper bound
on the cache overheads of the work-stealing parallel executions of a structured computation.

For structured computations with super final nodes,
it also makes sense to slightly relax the single-touch constraint as follows.

\begin{definition}
A structured single-touch computation with a \emph{super final node} is one where
each future thread $t$ at a fork $v$ has at least one and at most \emph{two} touches,
a descendant of $v$'s right child and the super final node.
\end{definition}

In such a computation,
a future thread can have the super final node as its only touch.
This structure corresponds to a program where one thread forks another thread to
accomplish a side-effect instead of computing a value.
The parent thread never touches the resulting future,
but the computation as a whole cannot terminate until the forked thread
completes its work.

Now we show that the parallel executions of structured single-touch computations
with super final nodes also have relatively low cache overheads.

\begin{lemma}\label{superfinalnode_lemmafutureparentfirst}
In the sequential execution of a structured single-touch computation with a super final node, where the future thread at a fork is always chosen to execute first,
any touch $x$'s future parent is executed before $x$'s local parent,
and the right child $u$ of any fork $v$ immediately follows the last node of
the future thread spawned at $v$, i.e.,
the future parent of the last touch of the future thread.
\end{lemma}

\begin{lemma}\label{superfinalnode_deviationtouch}
Let $t$ be the future thread at a fork $v$ in a structured single-touch computation with a super final node.
If a touch of $t$ or $v$'s right child $u$ is a deviation,
then either $u$ is stolen or there is a touch by $t$ which is a deviation.
\end{lemma}
\begin{proof}
The proofs of Lemma \ref{lemmafutureparentfirst} and Lemma~\ref{deviationtouch},
with only minor modifications, also apply to the above two lemmas, respectively.
That is because introducing the super final node into a computation doesn't affect
the order in which other nodes are executed, since no other nodes need to wait for
the super final node and the super final node is always the last node to execute.
More specifically, when a processor executing any thread $t$ reaches a node that
is a parent of the super final node, the processor will continue to work on $t$
if that node is not the last node of $t$, and otherwise try popping a node out of
its deque.
Therefore, by the same proof techniques as for
Lemmas \ref{lemmafutureparentfirst} and \ref{deviationtouch},
we can show that a processor will execute the right child $u$ of a fork $v$ and
the parents of the touches of the future spawned at $v$ in the order stated in
Lemmas \ref{superfinalnode_lemmafutureparentfirst} and \ref{superfinalnode_deviationtouch}.
\end{proof}

\begin{theorem}
If, at each fork, the future thread is chosen to execute first,
then a parallel execution with work stealing incurs $O(PT_{\infty}^2)$ deviations and
$O(CPT_{\infty}^2)$ additional cache misses in expectation on a structured single-touch
computation with a super final node.
\end{theorem}
\begin{proof}
The proof is similar to that of Theorem~\ref{futurefirstupperbound}. The only difference is that if a touch by a thread $t$ is a deviation, now the two touches of $t$ can both be deviations, which could be a trouble for constructing the deviation chains. Fortunately, one of these two touches is the super final node, which is always the last node to execute and hence will not make the touches of other threads become deviations. Therefore, we can still get a unique deviation chain starting from a steal and hence the proof of Theorem~\ref{futurefirstupperbound} still applies here.
\end{proof}

Similarly, we can also introduce a super final node to
a structured local-touch computation as follows.

\begin{definition}
A structured local-touch computation with a \emph{super final node} is one where
each future thread $t$ spawned at a fork $v$ can be touched only by the super final node
and by $t$'s parent thread at nodes that are descendants of the right child of $v$.
\end{definition}

It is obvious that by the same proof as for Theorem \ref{localtouchupperbound},
we can prove the following bounds.
\begin{theorem}\label{superfinalnode_localtouchupperbound}
If the future thread at a fork is always chosen to execute first,
then a parallel execution with work stealing incurs $O(PT_{\infty}^2)$ deviations and
$O(CPT_{\infty}^2)$ additional cache misses in expectation on a structured
local-touch computation with a super final node.
\end{theorem}
\section{Conclusions}
\label{section:Conclusions}
We have focused primarily on structured single-touch computations,
in which futures are used in a restricted way.
We saw that for such computations,
a parallel execution by a work-stealing scheduler that runs future threads
first can incur at most
$O(C P T^2_\infty)$ cache misses more than the corresponding sequential
execution,
a substantially better cache locality than the
$\Omega(C P T_\infty + C t T_\infty)$ worst-case
additional cache misses possible with unstructured use of futures.
Although we cannot prove this claim formally,
we think that these restrictions correspond to program structures that would
occur naturally anyway in many (but not all) parallel programs that use
futures.
For example,
Cilk~\cite{Blumofe95} programs are structured single-touch computations,
and that \cite{Blelloch97} observe
that the single-touch requirement substantially simplifies implementations.

We also considered some alternative restrictions on future use,
such as structured local-touch computations,
and structured computations with super final nodes,
that also incur a relatively low cache-locality penalty.
In terms of future work,
we think it would be promising to investigate how far these restrictions can be
weakened or modified while still avoiding a high cache-locality penalty.
We would also like to understand how these observations can be exploited by
future compilers and run-time systems.

% Bibliography
\bibliographystyle{abbrvnat}

\end{document}